\documentclass[runningheads]{llncs}

\usepackage[T1]{fontenc}
\usepackage{graphicx}
\graphicspath{{./figures/}} 
\usepackage{color}
\renewcommand{\figurename}{\bf{Figure}}

\spnewtheorem{observation}[theorem]{Observation}{\bfseries}{\itshape}
\spnewtheorem{ques}{Question}{\bfseries}{\itshape}

\usepackage{thmtools, thm-restate}
\usepackage{comment}
\usepackage{enumerate}
\usepackage{subcaption}
\usepackage[hidelinks]{hyperref}

\renewcommand{\subparagraph}[2]{\paragraph{\bf #2}}

\renewcommand{\orcidID}[1]{\href{https://orcid.org/#1}{\includegraphics[scale=.03]{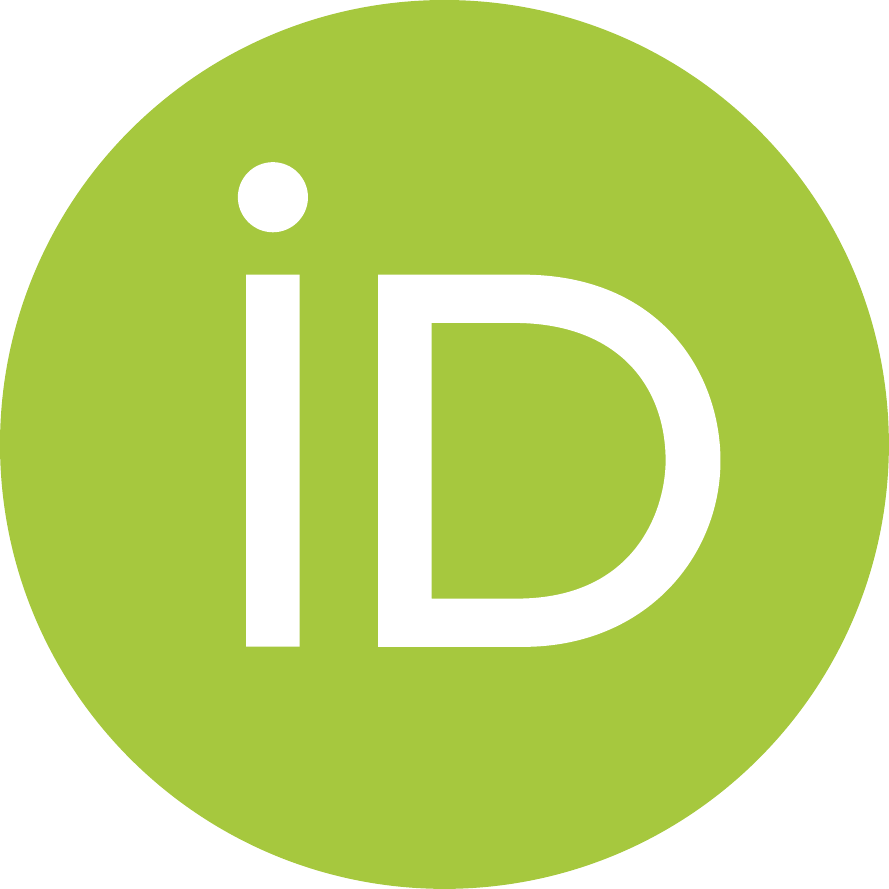}}}
\title{Flipping Plane Spanning Paths
\thanks{This work was initiated at the 2nd Austrian Computational Geometry Reunion Workshop in Strobl, June 2021. We thank all participants for fruitful discussions. J.~O.~is supported by ERC StG 757609.  O.~A. and R.~P.~are supported by FWF grant~W1230. B.~V. is supported by FWF Project \mbox{I 3340-N35}. K.K.~is supported by the German Science Foundation (DFG) within the research training group `Facets of Complexity' (GRK 2434). W.M.~is partially supported by the German Research Foundation within the collaborative DACH project \emph{Arrangements and Drawings} as DFG Project MU 3501/3-1, and by ERC StG 757609.}
}

\author{Oswin Aichholzer\inst{1}\orcidID{0000-0002-2364-0583} \and
Kristin Knorr\inst{2}\orcidID{0000-0003-4239-424X} \and
Wolfgang Mulzer\inst{2}\orcidID{0000-0002-1948-5840} \and \\
Johannes Obenaus\inst{2}\orcidID{0000-0002-0179-125X} \and
Rosna Paul\inst{1}\orcidID{0000-0002-2458-6427} \and
Birgit Vogtenhuber\inst{1}\orcidID{0000-0002-7166-4467}
}
\authorrunning{O. Aichholzer et al.}
\institute{Institute of Software Technology, Graz University of Technology, Austria\\
\email{\{oaich,bvogt,ropaul\}@ist.tugraz.at} \and
Institut für Informatik, Freie Universit{\"a}t Berlin, Germany\\
\email{\{kristin.knorr,wolfgang.mulzer,johannes.obenaus\}@fu-berlin.de}
}
\makeatletter
\let\orgdescriptionlabel\descriptionlabel
\renewcommand*{\descriptionlabel}[1]{%
  \let\orglabel\label
  \let\label\@gobble
  \phantomsection
  \edef\@currentlabel{#1\unskip}%
  \let\label\orglabel
  \orgdescriptionlabel{#1}%
}
\makeatother

\usepackage{linegoal}
\usepackage{wrapfig}
\usepackage{mathtools}
\usepackage{csquotes}

\newcommand{\PS}{\ensuremath{\mathcal{P}(S)}}
\newcommand{\GDC}{\ensuremath{\mathsf{GDC}}}

\newcommand{\HC}{\ensuremath{C}}  
\newcommand{\CC}{\ensuremath{CC}}  

\usepackage{xspace}

\newcommand{\mya}{\ensuremath{a}}
\newcommand{\myb}{\ensuremath{b}}
\newcommand{\myc}{\ensuremath{c}}
\newcommand{\myd}{\ensuremath{d}}

\newcommand{\weight}{weight\xspace}

\newcommand{\CH}{\operatorname{CH}}

\usepackage{todonotes}

\begin{document}

\maketitle

\begin{abstract}
Let $S$ be a planar point set in general position, and 
let $\PS$ be the set of all plane straight-line paths with 
vertex set $S$.
A flip on a path $P \in \PS$ is the operation
of replacing an edge $e$ of $P$ with another edge $f$ on $S$ to
obtain a new valid path from $\PS$.
It is a long-standing open question whether for every given point
set $S$, every path from~$\PS$ can be transformed into 
any other path from $\PS$ by a sequence of flips. 
To achieve a better understanding of this
question, we show that it is sufficient to prove the statement for 
plane spanning paths whose first edge is fixed. Furthermore, we
provide positive answers for special classes of point sets, namely, for 
wheel sets and generalized double circles
(which include, e.g., double chains and double circles).

\end{abstract}

\section{Introduction}

Reconfiguration is a classical and widely studied topic with various applications in multiple areas. A natural way to provide structure for a reconfiguration problem is by studying the so-called \emph{flip graph}. For a class of objects, the flip graph has a vertex for each element and adjacencies are determined by a local flip operation (we will give the precise definition shortly). 
In this paper we are concerned with transforming plane spanning paths via edge flips. 

Let $S$ be a set of $n$ points in the plane in general position 
(i.e., no three points are collinear), and 
let $\PS$ be the set of all plane straight-line 
spanning paths for~$S$, i.e., the set of all paths with vertex set
$S$ whose straight-line embedding 
on~$S$ is crossing-free. 
A \emph {flip} on a path $P \in \PS$ is the operation of
removing an edge $e$ from $P$ and replacing it with another
edge $f$ on $S$ such that
the graph $(P \setminus e) \cup f$ is again a valid path from $\PS$ 
(note that the edges $e$ and~$f$ might cross). The \emph{flip graph}
on $\PS$ has vertex set $\PS$ and two vertices are adjacent if and only if the corresponding paths differ by a single flip.
The following conjecture will be the focus of this paper:

\begin{conjecture}[Akl et al.~\cite{AKL2007}]\label{conj:main}
For every point set $S$ in general position, the flip graph on $\PS$ is connected.
\end{conjecture}

\subparagraph*{Related work.} 

For further details on reconfiguration problems in general we refer the reader to the surveys of Nishimura~\cite{nishimura2018} and Bose and Hurtado~\cite{survey}.
Connectivity properties of flip graphs have been studied extensively in a huge variety of settings, see, e.g., \cite{hernando2002graphs,houle,lawson1972transforming,nichols2020,Wagner1936} for results on triangulations, matchings and trees.

In our setting of plane spanning paths, flips are much more restricted, making it more difficult to prove a positive answer. Prior to our work only results for 
point sets in convex position and very small point sets were known. Akl et al.~\cite{AKL2007}, who initiated the study of flip connectivity on plane spanning paths, showed connectedness of the flip graph on $\PS$ if $S$ is in convex position or $|S| \leq 8$. In this convex setting, Chang and Wu~\cite{CW2009} derived tight bounds concerning the diameter of the flip graph, namely, $2n - 5$ for $n = 3, 4$, and $2n - 6$ for $n \geq 5$.

For the remainder of this paper, we consider the flip graph on $\PS$ (or a subset of $\PS$). Moreover, unless stated otherwise, the word \emph{path} always refers
to a path from $\PS$ for an underlying point set
$S$ that is clear from the context. 

\begin{figure}[t]
	\centering
	{\includegraphics[page=1,scale=0.55]{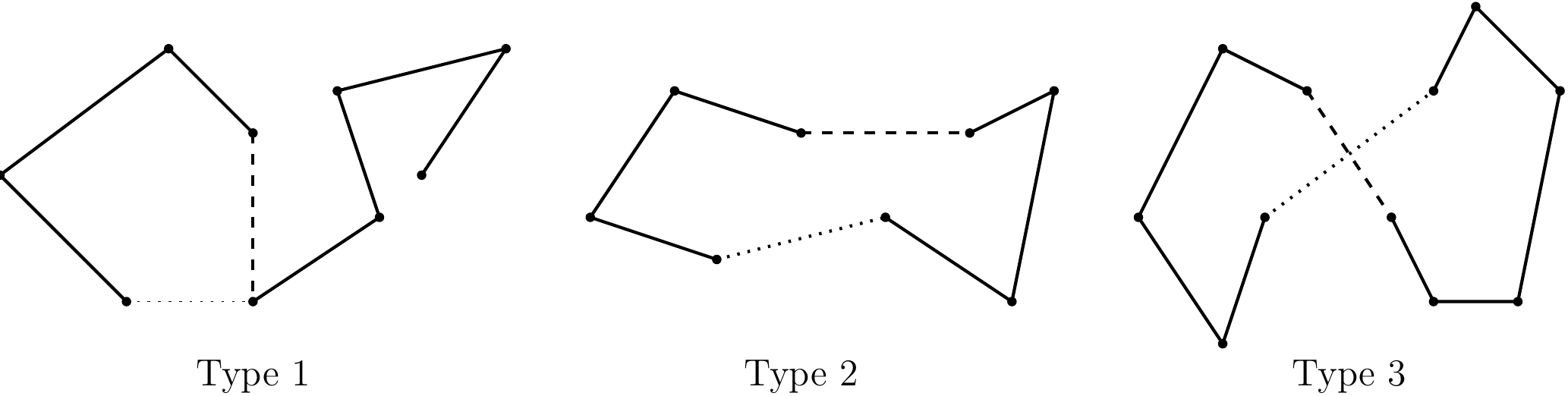}}
\caption{\label{fig:fliptype}
	The three types of flips in plane spanning paths.
	}
\end{figure}

\subparagraph*{Flips in plane spanning paths.}
Let us have a closer look at the different types of possible flips 
for a path $P~=~v_1, \dots, v_n~\in~\PS$
(see also Figure~\ref{fig:fliptype}). 
When removing an edge $v_{i - 1}v_i$ from $P$ with $2 \leq i \leq n$, 
there are three possible
new edges that can be added in order to obtain a path (where, of course,
not all three choices will necessarily lead to a valid path in $\PS$):
$v_1v_i$, $v_{i - 1}v_n$, and~$v_1v_n$. 
A flip of \emph{Type 1} is a valid flip that adds the edge
$v_1v_i$ (if $i > 2$) or the edge $v_{i-1}v_n$ (if $i < n$).
It results in 
the path $v_{i - 1}, \dots, v_1, v_i, \dots, v_n$, or the path
$v_1, \dots, v_{i - 1}, v_n, \dots, v_i$. That is, a 
Type 1 flip inverts a contiguous chunk from one of the two ends of the path.
A flip of \emph{Type 2} adds the edge $v_1v_n$ and has the additional property
that the edges $v_{i - 1}v_i$ and $v_1v_n$ do not cross.
In this case, the path $P$ together with the edge $v_1v_n$ forms a plane cycle.
If a Type 2 flip is possible for one 
edge $v_{i - 1}v_i$ of $P$, then it is possible for all edges of $P$.
A Type 2 flip can be simulated by a sequence of Type 1 flips, e.g., flip
$v_1v_2$ to $v_1v_n$, then flip $v_2v_3$ to $v_1v_2$, then $v_3v_4$ to
$v_2v_3$, etc., until flipping $v_{i - 1}v_{i}$ to $v_{i - 2}v_{i - 1}$.
A flip of \emph{Type 3} also adds the edge $v_1v_n$, but now the
edges $v_1v_n$ and $v_{i - 1}v_i$ cross. 
Note that a Type 3 flip is 
only possible
if the edge $v_1v_n$ crosses exactly one edge of $P$, and then the flip
is possible only for the edge $v_{i - 1}v_i$ that is crossed.

\subparagraph*{Contribution.} 
We approach Conjecture~\ref{conj:main} from two directions. 
First, we show that it is sufficient to prove flip connectivity for paths with a fixed starting edge. 
Second, we verify Conjecture~\ref{conj:main} for several classes of point sets, namely wheel sets and generalized double circles (which include, e.g., double chains and double circles).

Towards the first part, we define, 
for two distinct points $p, q \in S$, the following subsets of $\PS$: 
let $\mathcal{P}(S, p)$ 
be the set of all plane spanning paths for $S$ that start at $p$, 
and let $\mathcal{P}(S, p, q)$ be the set of all plane spanning paths 
for $S$ that start at $p$ and continue with $q$. 
Then for any $S$, the flip graph on $\mathcal{P}(S, p, q)$ is 
a subgraph of the flip graph on $\mathcal{P}(S, p)$, which in turn is 
a subgraph of the flip graph on~$\PS$. 
We conjecture that all these flip graphs are connected:

\begin{conjecture}\label{conj:startpoint}
For every point set $S$ in general position and every $p \in S$, the flip graph on $\mathcal{P}(S,p)$ is connected.
\end{conjecture}

\begin{conjecture}\label{conj:startedge}
For every point set $S$ in general position and every $p,q \in S$, the flip graph on $\mathcal{P}(S,p,q)$ is connected.
\end{conjecture}

Towards Conjecture~\ref{conj:main}, we show that it suffices to prove Conjecture~\ref{conj:startedge}:

\begin{restatable}{theorem}{lemImplicationTwoOne}\label{lem:implication_2_1}
Conjecture~\ref{conj:startpoint} implies Conjecture~\ref{conj:main}.
\end{restatable}

\begin{restatable}{theorem}{lemImplicationThreeTwo}\label{lem:implication_3_2}
Conjecture~\ref{conj:startedge} implies Conjecture~\ref{conj:startpoint}.
\end{restatable}

\begin{figure}[t]
	\centering
	{\includegraphics[page=2,scale=0.7]{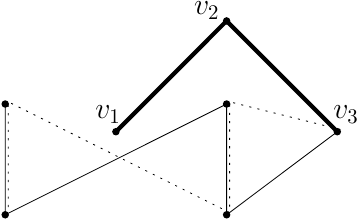}}
\caption{\label{fig:threevertexpath}
	Example where the flip graph is disconnected 
	if the first three vertices of the paths are fixed. 
	No edge of the solid path can be
flipped, but there is at 
	least one other path (dotted) with the same three starting vertices.
	}
\end{figure}

Note that the analogue of
Conjecture~\ref{conj:startedge} for paths where the first $k \geq 3$ vertices are fixed,
does not hold:
Figure~\ref{fig:threevertexpath} shows a counterexample with $7$ 
points and $k=3$. 

\begin{wrapfigure}{r}{2.8cm}
\includegraphics[page=1,scale=0.67]{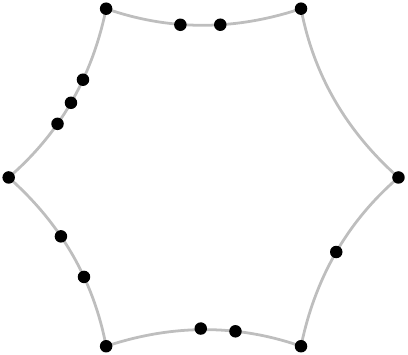}
\end{wrapfigure}
Towards the flip connectivity for special classes of point sets, we consider wheel sets and generalized double circles. A point set is in \emph{wheel configuration} if it has exactly one point inside the convex hull. For generalized double circles we defer the precise definition to Section~\ref{sec:GDC}, however, intuitively speaking a generalized double circle is obtained by replacing
each edge of the convex hull by a flat enough concave chain of
arbitrary size (as depicted on the right). 
We show that the flip graph is connected in both cases:

\begin{restatable}{theorem}{thmWheel}\label{thm:wheel}
Let $S$ be a set of $n$ points in wheel configuration. Then the flip graph (on $\PS$) is connected with diameter at most $2n-4$.
\end{restatable}

\begin{restatable}{theorem}{thmGDC}\label{thm:GDC}
	Let $S$ be a set of $n$ points in generalized double circle configuration. Then the flip graph (on $\PS$) is connected with diameter at most~$O(n^2)$.
\end{restatable}

Finally, we remark that using the order type database~\cite{abcd}, we are able to computationally verify Conjecture~\ref{conj:main} for every set of $n \leq 10$ points 
in general position (even when using only Type 1 flips).\footnote{The source code is available at \url{https://github.com/jogo23/flipping_plane_spanning_paths}.}

\subparagraph*{Notation.} We denote the convex hull of a point set $S$ by $\CH(S)$. All points $p \in S$ on the boundary of $\CH(S)$ are called \emph{extreme points} and the remaining points are called \emph{interior} points.

\section{A Sufficient Condition}\label{sec:main_fixed_edge}

In this section we prove Theorem~\ref{lem:implication_2_1} and Theorem~\ref{lem:implication_3_2}.

\begin{restatable}{lemma}{lemPathexists}\label{lem:pathexists}
Let $S$ be a point set in general position and $p,q\in S$. Then there exists a path $P \in {\cal P}(S)$ which has $p$ and $q$ as its end vertices.
\end{restatable}

\begin{proof}
We consider two different cases. Let us first assume that at least one of the two points, w.l.o.g.\ say $p$, is not an extreme point of $S$. Sort all other points of $S$ radially around $p$, starting at $q$. Connect $p$ to the second point in this order (the point radially just after $q$) and connect all other points in radial order to a path, such that $q$ becomes the last point of this path; see Figure~\ref{fig:twovertexpath} (left). Since $p$ is an interior point, each edge of this path stays in a cone defined by $p$ and two successive points in the radial order, in particular all these cones are disjoint. Hence, we have obtained a plane spanning path starting at $p$ and ending at $q$.

Assume now for the second case that both given points lie on the boundary of the convex hull of $S$. Then consider two tangents of the convex hull of $S$ going through $p$ and $q$, respectively. By the general position assumption of $S$ we can always perturb these two tangents so that they still go through the two given vertices but are not parallel, and thus cross in a point $x$ outside of the convex hull of $S$; see Figure~\ref{fig:twovertexpath} (right). Sort all points radially around $x$ and connect the points in this order to a path. By construction the points $p$ and $q$ are the first and last point in this sorting, and we thus obtain the required path.\qed
\end{proof}

\begin{figure}[t]
	\centering
	{\includegraphics[page=1,scale=0.55]{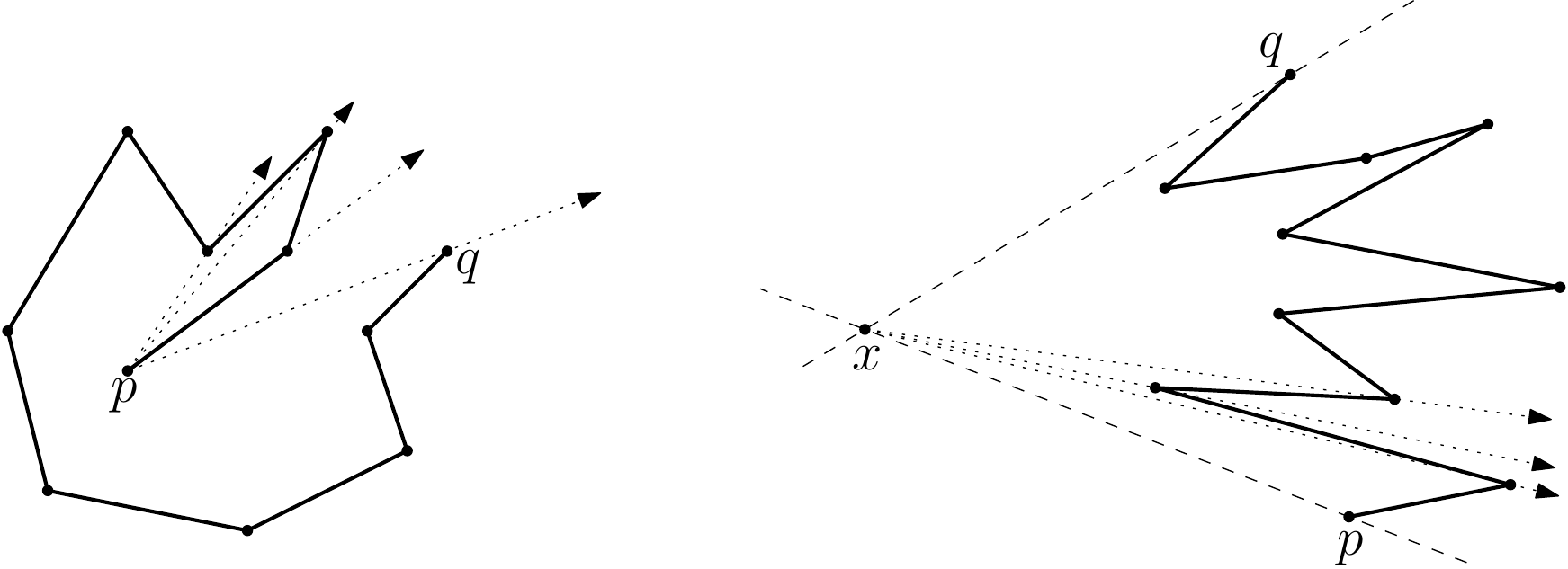}}
\caption{\label{fig:twovertexpath}
	For any two given points $p$ and $q$ there exists a plane spanning path having these two points as start and target points. Left: the case if at least one of the two points is in the interior of the point set. Right: the case when both points are on the boundary of the convex hull.}
\end{figure}

\lemImplicationTwoOne*

\begin{proof}
Let $S$ be a point set and $P_s, P_t \in \PS$. If $P_s$ and $P_t$ have a common endpoint, we can directly apply Conjecture~\ref{conj:startpoint} and the statement follows. So assume that $P_s$ has the endpoints $v_a$ and $v_b$, and $P_t$ has the endpoints $v_c$ and $v_d$, which are all distinct. By Lemma~\ref{lem:pathexists} there exists a path $P_m$ having the two endpoints $v_a$ and~$v_c$.  By Conjecture~\ref{conj:startpoint} there is a flip sequence from $P_s$ to $P_m$ with the common endpoint $v_a$, and again by Conjecture~\ref{conj:startpoint} there is a further flip sequence from $P_m$ to $P_t$ with the common endpoint $v_c$. 
This concludes the proof.\qed%
\end{proof}

Towards Theorem~\ref{lem:implication_3_2}, we first have a closer look at what edges form \emph{viable} starting edges. For a given point set $S$ and points $p,q \in S$, we say that $pq$ forms a \emph{viable} starting edge if there exists a path $P \in \PS$ that starts with $pq$. For instance, an edge connecting two extreme points that are not consecutive along $\CH(S)$ is not a viable starting edge. The following lemma shows that these are the only non-viable starting edges.

\begin{restatable}{lemma}{lemViableStartEdge}\label{lem:viable_start_edge}
Let $S$ be a point set in general position and $u,v \in S$. The edge $uv$ is a viable starting edge if and only if one of the following is fulfilled:
(i) $u$ or $v$ lie in the interior of $\CH(S)$, or (ii) $u$ and $v$ are consecutive along $\CH(S)$.
\end{restatable}

\begin{proof}
If $u$ and $v$ are extreme points that are not consecutive along $\CH(S)$, then there exist points to the left and to the right of the edge $uv$. However, any plane path starting with $uv$ can only reach the points either to the left or the right.

If $u$ is in the interior of $\CH(S)$, we sort the remaining points in radial order around $u$. We construct a path starting with $uv$ that visits the remaining points consecutively in this radial order. We proceed exactly the same way, if $u$ and $v$ are consecutive extreme points.

If $u$ is an extreme point and $v$ an interior point, let $v_\ell$ and $v_r$ be the two neighbors of $u$ along the convex hull. Let $S' \subset S$ be the set of points in the interior of $\CH(S)$ plus $\{v_\ell, v_r\}$. Again we sort the points of $S'$ in radial order around $u$. Let $v^- \in S'$ be the vertex that is right before $v$ in this radial order.

We construct three paths $P_1 = v_\ell, \ldots, v^-$, $P_2 = v_r, \ldots, v$, and $P_3 = v_\ell, \ldots, v_r$, where $P_1$ and $P_2$ simply connect the points in $S'$ in radial order (between its corresponding endpoints). Note that $P_1$ may have length zero. $P_3$ connects $v_\ell$ to $v_r$ along the boundary of the convex hull of $S$ (along the side not containing $v_1$). Then the union of the three paths $P_1, P_2, P_3$ together with $uv$ 
forms the desired plane spanning path (see Figure~\ref{fig:hamcycle2} (right)).\qed
\end{proof}

\begin{figure}[t]
	\centering
	{\includegraphics[page=4,scale=0.65]{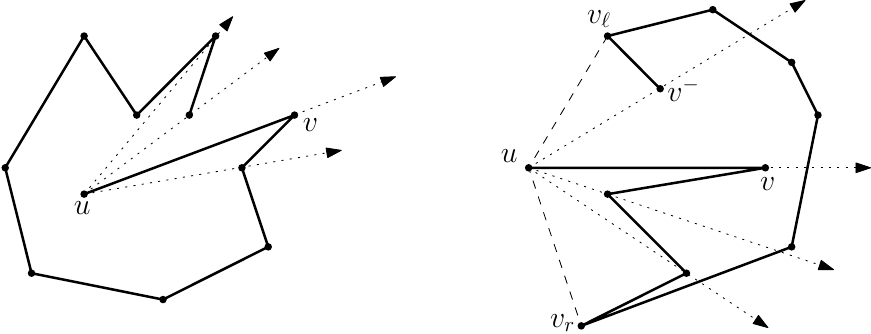}}
\caption{\label{fig:hamcycle2}
	\emph{Left:} $u$ lies inside the convex hull of $S$. \emph{Right:} $u$ is on the boundary of the convex hull. For paths $P_1$ and $P_2$: only vertices in the radial order around $u$ which are not on the boundary of 
the convex hull of $S$, are considered.}
\end{figure}

The following lemma is the analogue of Lemma~\ref{lem:pathexists}:

\begin{restatable}{lemma}{lemCycle}\label{lem:cycle}
Let $S$ be a point set in general position and $v_1 \in S$. Further let $S' \subset S$ be the set of all points $p \in S$ such that $v_1p$ forms a viable starting edge. Then for two points $q,r \in S'$ that are consecutive in the circular order around~$v_1$, there exists a plane spanning cycle containing the edges $v_1q$ and $v_1r$.
\end{restatable}

\begin{proof}\label{lem:cycle:proof}
The construction of these plane spanning cycles is completely analogous to the construction of the paths in the proof of Lemma~\ref{lem:viable_start_edge} but we add the proof for the sake of completeness.

First assume that $v_1$ is an interior point. Then, by Lemma~\ref{lem:viable_start_edge}, $S' = S \setminus \{v_1\}$ holds. We construct a plane spanning path starting with $v_1q$ and visiting the remaining points in circular order around $v_1$ such that $r$ is the last in this order. Lastly, connect $r$ to $v_1$. 

Now, let $v_1$ be an extreme point. Again we proceed analogously if $q$ and $r$ are the two neighbors of $v_1$ along $\CH(S)$.
Otherwise, by Lemma~\ref{lem:viable_start_edge}, at least one of the vertices $q$ or $r$ is an interior point. Then we construct the same three paths $P_1, P_2, P_3$ as in Lemma~\ref{lem:viable_start_edge} (replacing the roles of $v, v^-$ by $q$ and $r$). Then the union of $P_1, P_2, P_3$ together with $v_1q$ and $v_1r$ forms the desired cycle (see Figure~\ref{fig:hamcycle} (right)).\qed
\end{proof}

\begin{figure}[t]
	\centering
	{\includegraphics[page=3,scale=0.65]{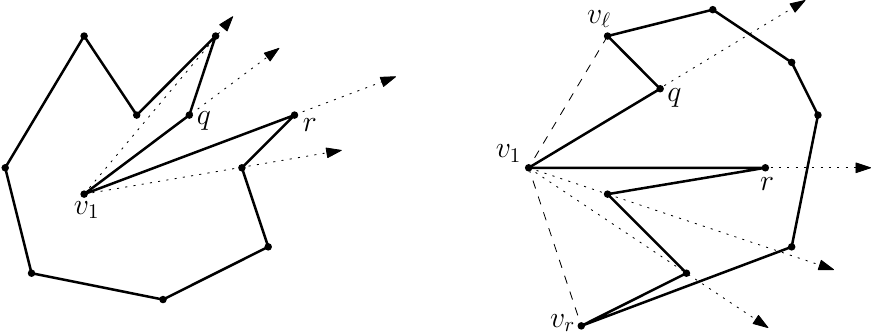}}
\caption{\label{fig:hamcycle}
	There exists a plane spanning cycle of $S$ such that $v_1$ is connected to two points neighbored in the radial order around $v_1$. \emph{Left:} $v_1$ is an interior point. \emph{Right:} $v_1$ is an extreme point.
}
\end{figure}

We are now ready to prove Theorem~\ref{lem:implication_3_2}:

\lemImplicationThreeTwo*

\begin{proof}
Let $S$ be a point set and $v_1 \in S$. Further let $P, P' \in \mathcal{P}(S,v_1)$.
If $P$ and $P'$ have the starting edge in common, then we directly apply Conjecture~\ref{conj:startedge} and are done. So let us assume that the starting edge of $P$ is $v_1v_2$ and the starting edge of $P'$ is $v_1v_2'$.
Clearly $v_2, v_2' \in S'$ holds. Sort the points in $S'$ in radial order around $v_1$. Further let $v_x \in S'$ be the next vertex after $v_2$ in this radial order and $C$ be the plane spanning cycle with edges $v_1v_2$ and $v_1v_x$, as guaranteed by Lemma~\ref{lem:cycle}.

By Conjecture~\ref{conj:startedge}, we can flip $P$ to $C \setminus v_1v_x$. Then, flipping $v_1v_2$ to $v_1v_x$ we get to the path $C \setminus v_1v_2$, which now has $v_1v_x$ as starting edge. We iteratively continue this process of \enquote{rotating} the starting edge until reaching $v_1v_2'$.\qed
\end{proof}

Theorems~\ref{lem:implication_2_1} and \ref{lem:implication_3_2} imply that it suffices to show connectedness of certain subgraphs of the flip graph. A priori it is not clear whether this is an easier or a more difficult task -- on the one hand we have smaller graphs, making it easier to handle. On the other hand, we may be more restricted concerning which flips we can perform, or exclude certain \enquote{nice} paths.

\section{Flip Connectivity for Wheel Sets}\label{sec:wheel}

Akl et al.~\cite{AKL2007} proved connectedness of the flip graph if the underlying point set~$S$ is in convex position. They showed that every path in $\PS$ can be flipped to a \emph{canonical path} that uses only edges on the convex hull of $S$. To generalize this approach to other classes of point sets, we need two ingredients: (i) a set of \emph{canonical paths} that serve as the target of the flip operations and that have the property that any canonical path can be transformed into any other canonical path by a simple sequence of flips, usually of constant length; and (ii) a strategy to flip any given path to some canonical path. 

Recall that a set $S$ of $n \geq 4$ points
in the plane is a \emph{wheel set} if there is exactly one interior point $c_0 \in S$. 
We call $c_0$ the \emph{center} of $S$ 
and classify the edges on $S$ as follows:
an edge incident to the center $c_0$ is called a \emph{radial} edge,
and an edge along $\CH(S)$ is called \emph{spine} edge 
(the set of spine edges forms the \emph{spine}, which is just the boundary of the convex hull here). All other edges are called \emph{inner} edges.
The \emph{canonical paths} are those that consist only of spine edges and one or two radial edges. 

We need two observations. Let $S$ be a point set and $P = v_1, \ldots, v_n \in \PS$. 
Further, let $v_i$ $(i \geq 3)$ be a vertex such that no edge on $S$ crosses  $v_1v_i$. 
We denote the face bounded by $v_1, \ldots, v_i, v_1$ by $\Phi(v_i)$.

\begin{observation}\label{obs:face}
Let $S$ be a point set, $P = v_1, \dots, v_n \in \PS$, and $v_i$ ($i~\geq~3$) be a vertex such that no edge on $S$ crosses $v_1v_i$. 
Then all vertices after $v_i$ (i.e.,~$\{v_{i+1}, \ldots, v_n\}$) must entirely be contained in either the interior or the exterior of $\Phi(v_i)$.
\end{observation}

\begin{observation}\label{app:obs:diagonal}
Let $S$ be a wheel set and
let $P =  v_1,\dots,v_n \in \mathcal{P}(S)$.
Suppose that the edge $v_iv_{i+1}$ of $P$ is an inner edge. Then 
the sets $\{v_1,\dots, v_{i-1}\}$ and $\{v_{i+2},\dots, v_n\}$ lie 
on different sides of the line spanned by $v_iv_{i+1}$. 
\end{observation}

\thmWheel*

\begin{proof}
Let $P = v_1, \dots, v_n \in \PS$ be a non-canonical path. W.l.o.g., we can assume $v_1 \neq c_0$ (at least one of the two endpoints of $P$ is not the center).
We show how to apply suitable flips to increase the number of spine edges of $P$.

By Lemma~\ref{lem:viable_start_edge}, the edge $v_1v_2$ is 
either radial or a spine edge. We distinguish the two cases:

\begin{description}
\item[Case 1\label{wheel_case_1}] $v_1v_2$ is radial, i.e., (since we assumed
 $v_1 \neq c_0$), we have $v_2 = c_0$.

Then $v_2v_3$ is radial and, in analogy to Lemma~\ref{lem:viable_start_edge}, the remaining path can only visit vertices on one side of $v_1v_2v_3$ (see Figure~\ref{app:fig:wheel} (left)). Hence, $v_3$ must be a neighbor of $v_1$ along the convex hull.

Thus 
we can increase the number of spine edges by flipping the radial edge 
$v_2v_3$ to the spine edge $v_1v_3$.

\item[Case 2\label{wheel_case_2}] $v_1v_2$ is a spine edge.

Let $v_\mya$ ($\mya \neq 2$) be a neighbor of $v_1$ (along the convex hull). Note first that we can assume $v_{\mya-1}v_\mya$ to be a spine edge, since otherwise we can increase the number of spine edges by flipping $v_{\mya-1}v_\mya$ to $v_1v_\mya$. Furthermore, we can assume $a \neq n$, since otherwise we can insert $v_1v_n$ and remove an arbitrary (non-spine) edge.

By Observation~\ref{app:obs:diagonal}, $P$ cannot have any inner edge $e$ 
before $v_\mya$, since otherwise $e$ would have the neighbors $v_1$ and $v_\mya$ on 
the same side. 
On the other hand, since $v_1v_2$ and $v_{\mya-1}v_\mya$ are spine edges, $v_{\mya+1}$ must be in the interior of the face $\Phi(v_\mya)$. 
Then, by Observation~\ref{obs:face}, all extreme points must already be covered before $v_{\mya+1}$ (see Figure~\ref{app:fig:wheel} (right)).
This, however, implies that $P$ does not contain an inner edge.

Hence, the case that $v_1v_2$ (and $v_{\mya-1}v_\mya$) are spine edges cannot occur, since we assumed $P$ 
to be non-canonical, i.e., containing an inner edge.
\end{description}

\begin{figure}[t]
	\centering
	{\includegraphics[page=1,scale=0.75]{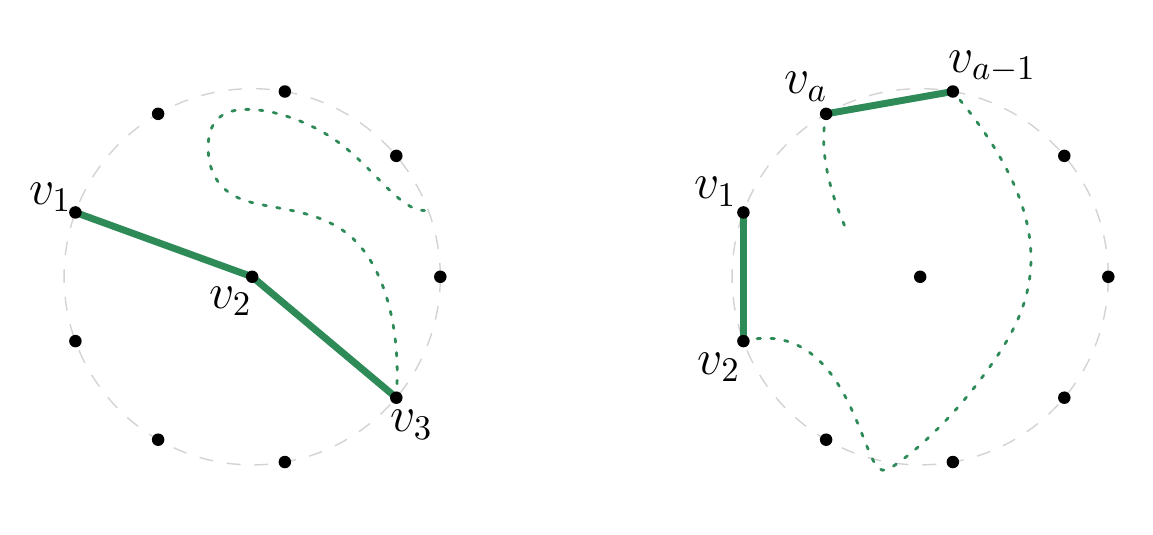}}
\caption{\label{app:fig:wheel}
	\emph{Left:} Illustration of \ref{wheel_case_1}: if $v_1$ and $v_3$ are not neighbors, there are vertices that cannot be reached. \emph{Right:} Illustration of \ref{wheel_case_2}: if $v_{\mya-1}v_\mya$ is a spine edge, the path after $v_\mya$ cannot visit any extreme point anymore.}
\end{figure}

Iteratively applying above procedure, we obtain a canonical path. Any two such canonical paths can be transformed into each other by at most 4 flips.

By Observation~\ref{app:obs:diagonal} and the fact that not all edges can be radial, we know that any plane spanning path contains at least one spine edge and hence, we need at most $n-3$ iterations. In total, we need 
at most $2(n - 3) + 4 = 2n - 2$ flips. However, if we count a 
little more carefully, we see that in the case of $n - 3$ iterations, i.e., 
only one radial edge remains and not two, we 
actually need only two flips instead of 4 
for the intermediate step. Hence, we save two steps.\qed
\end{proof}

\section{Flip Connectivity for Generalized Double Circles}\label{sec:GDC}

The proof for generalized double circles is in principle similar to the one for wheel sets but much more involved. 
For a point set $S$ and two extreme points $p,q \in S$, we call a subset $\CC(p,q) \subset S$ \emph{concave chain} (chain for short) for $S$, if (i) $p,q \in \CC(p,q)$; (ii) $\CC(p,q)$ is in convex position; (iii) $\CC(p,q)$ contains no other extreme points of $S$; and 
(iv) every line $\ell_{xy}$ through any two points $x,y \in \CC(p,q)$ has the property that all points of $S \setminus \CC(p,q)$ are contained in the open halfplane bounded by $\ell_{xy}$ that contains neither $p$ nor $q$. 
Note that the extreme points $p$ and $q$ must necessarily be consecutive along $\CH(S)$. 
If there is no danger of confusion, we also refer to the spanning path from $p$ to $q$ along the convex hull of $\CC(p,q)$ as the concave chain.

\begin{figure}[t]
	\centering
	{\includegraphics[page=2,scale=0.62]{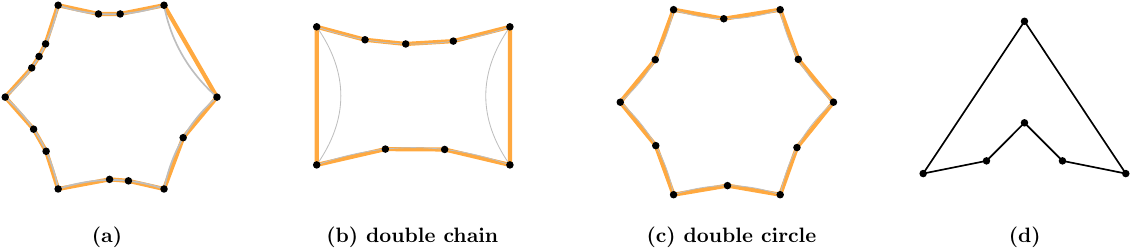}}
\caption{\label{fig:gdc2}
	(a-c) Examples of generalized double circles (the uncrossed spanning cycle is depicted in orange). 
	(d) A point set admitting an uncrossed spanning cycle that is \emph{not} a generalized double circle.
	}
\end{figure}

A point set $S$ is in \emph{generalized double circle} position if there exists a family of concave chains such that every inner point of $S$ is contained in exactly one chain and every extreme point of $S$ is contained in exactly two chains.
We denote the class of generalized double circles by $\GDC$ (see Figure~\ref{fig:gdc2} for an illustration). For $S \in \GDC$, it is
 not hard to see that the union of the concave chains forms an uncrossed spanning cycle: 

\begin{lemma}\label{lem:gdc_fulfills_p1}
Every point set $S \in \GDC$ admits an uncrossed spanning cycle formed by the union of the concave chains.
\end{lemma}

\begin{proof}
Let $S \in \GDC$ and denote the extreme points of $S$ by $u_1, \dots, u_k$ in circular order. Assume, for the sake of contradiction, that there is an edge $rs$ along the convex hull of a concave chain $\CC(u_i, u_{i+1})$ that is crossed by some edge $pq$. Then $p$ and $q$ lie on different sides of the line $\ell_{rs}$ through $rs$ and hence, at least one of the points $p$ or $q$ must belong to $\CC(u_i, u_{i+1})$, say $q \in \CC(u_i, u_{i+1})$. Furthermore, $\ell_{qs}$ or $\ell_{qr}$ must 
have $p$ on the same side as $u_i$ or $u_{i+1}$ (note that this is also the case if $r$ or $s$ and $q$ coincide with $u_i, u_{i+1}$); see Figure~\ref{app:fig:gdc_uncrossed_cycle}. Hence, also $p \in \CC(u_i, u_{i+1})$ holds. However, since $rs$ is a hull edge of $\CC(u_i, u_{i+1})$ it cannot be crossed by $pq$; a contradiction.\qed
\end{proof}

\begin{figure}[t]
\centering
\includegraphics[scale=0.7,page=1]{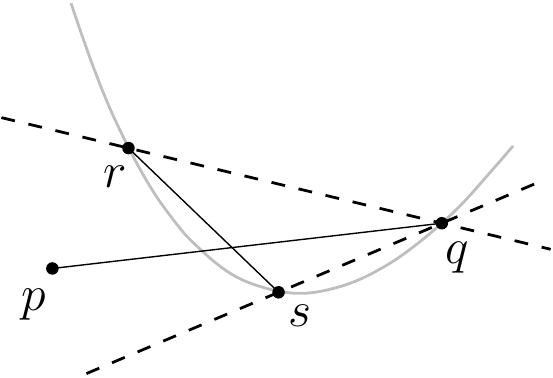}
\caption{Illustration of the proof of Lemma~\ref{lem:gdc_fulfills_p1}.}
\label{app:fig:gdc_uncrossed_cycle}
\end{figure}

Before diving into the details of the proof for generalized double circles, we start by collecting preliminary results in a slightly more general setting,
namely for point sets $S$ fulfilling the following property: 

\begin{description}
\item[(P1)\label{item:p1}] there is an \emph{uncrossed} 
spanning cycle $\HC$ on $S$, i.e.,  
no edge joining two points of $S$ crosses any edge~of~$\HC$.
\end{description}

A point set fulfilling \ref{item:p1} is called \emph{spinal} point set. When considering a spinal point set $S$, we first fix an uncrossed spanning cycle~$\HC$, which we call \emph{spine} and all edges in $\HC$ \emph{spine edges}. For instance, generalized double circles are spinal point sets and the spine is precisely the uncrossed spanning cycle formed by the concave chains as described above.
Whenever speaking of the spine or spine edges for some point set without further specification, the underlying uncrossed cycle is either clear from the context, or the statement holds for any choice of such a cycle. Furthermore, we call all edges in the exterior/interior of the spine \emph{outer/inner edges}.

We define the \emph{canonical paths} to be those that consist only of spine edges.
Note that this definition also captures the
canonical paths used by Akl et al., and that any canonical path
can be transformed into any other by a single flip (of Type~2).
Two vertices incident to a common spine edge are called \emph{neighbors}.

Next, we need a strategy to flip an arbitrary path to a canonical path.
The biggest issue is to ensure that all intermediate
paths in the flip sequence are plane.
To this end, it is very helpful if many triangles that are spanned
by a spine edge 
and a third point of $S$ are \emph{empty}, i.e., 
do not contain any other point from $S$ in their interior. 
In the extreme case, where all such triangles are empty 
it is not hard to show that the flip graph is connected. We omit the details 
as it turns out that only convex point sets fulfill this extreme case (see Appendix~\ref{app:sec:extreme_case}).

\subparagraph*{Valid flips.}
We collect a few observations which will be useful to confirm the validity of a flip. Whenever we apply more than one flip, the notation in subsequent flips refers to the original graph and not the current (usually we apply one or two flips in a certain step). 
Figure~\ref{fig:valid_flips} gives an illustration of Observation~\ref{obs:valid_flips}.

\begin{figure}[t]
	\centering
	{\includegraphics[page=1,scale=0.55]{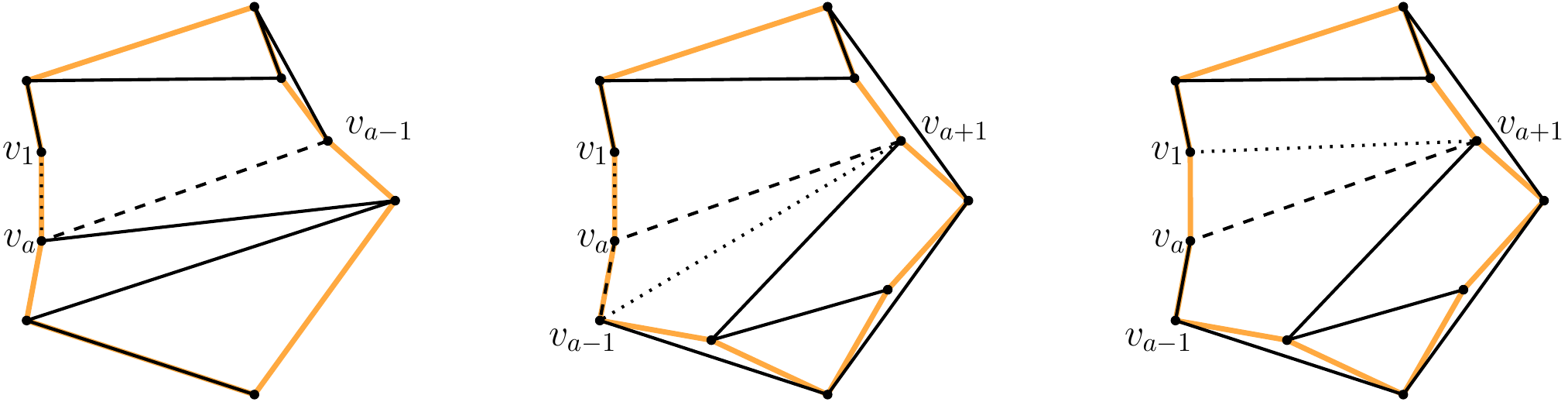}}
\caption{\label{fig:valid_flips}
	\emph{Left to right:} Illustration of the three flips in Observation~\ref{obs:valid_flips}. The spine is depicted in orange and edge flips are indicated by replacing dashed edges for dotted (in the middle, the two flips must of course be executed one after the other).
	}
\end{figure}

\begin{observation}\label{obs:valid_flips}
Let $S$ be a spinal point set, $P = v_1, \dots, v_n \in \PS$, and $v_1, v_\mya$ ($\mya \neq 2$) be neighbors. Then the following flips are valid (under the specified additional assumptions):

\vspace*{-0.35cm}
\begin{align*}
(a&) \text{ flip } v_{\mya-1}v_\mya \text{ to } v_1v_\mya \\[1.1em]
(b&) \text{ flip } v_\mya v_{\mya+1} \text{ to } v_{\mya-1}v_{\mya+1} && \parbox[t]{\linegoal}{ (if the triangle $\Delta v_{\mya-1}v_\mya v_{\mya+1}$ is empty and (b) is \\ performed subsequently after the flip in (a)) } \\[0.1em]
(c&) \text{ flip } v_\mya v_{\mya+1} \text{ to } v_1v_{\mya+1} && \parbox[t]{\linegoal}{ (if the triangle $\Delta v_1v_\mya v_{\mya+1}$ is empty and \\ $v_{\mya-1}v_\mya$ is a spine edge) }
\end{align*}
\end{observation}

Strictly speaking, in Observation~\ref{obs:valid_flips}(c) we do not require $v_{\mya-1}v_\mya$ to be a spine edge, but merely to be an edge not crossing $v_1v_{\mya + 1}$. 
The following lemma provides structural properties for generalized double circles, if the triangles in Observation~\ref{obs:valid_flips}(b,c) are non-empty: 

\begin{restatable}{lemma}{lemPrelGDCOne}\label{app:lem:prel_gdc_1}
Let $S \in \GDC$ and $p,q,x \in S$ such that $p$ and $q$ are neighbors. Further, let the triangle $\Delta pqx$ be non-empty. Then the following holds:
\begin{enumerate}[(i)]
\item At least one of the two points $p,q$ is an extreme point (say $p$),
\item $x$ does not lie on a common chain with $p$ \emph{and} $q$, but shares a common chain with either $p$ or $q$ (the latter may only happen if $q$ is also an extreme point).
\end{enumerate}
\end{restatable}

\begin{figure}[t]
\centering
\includegraphics[scale=0.63,page=1]{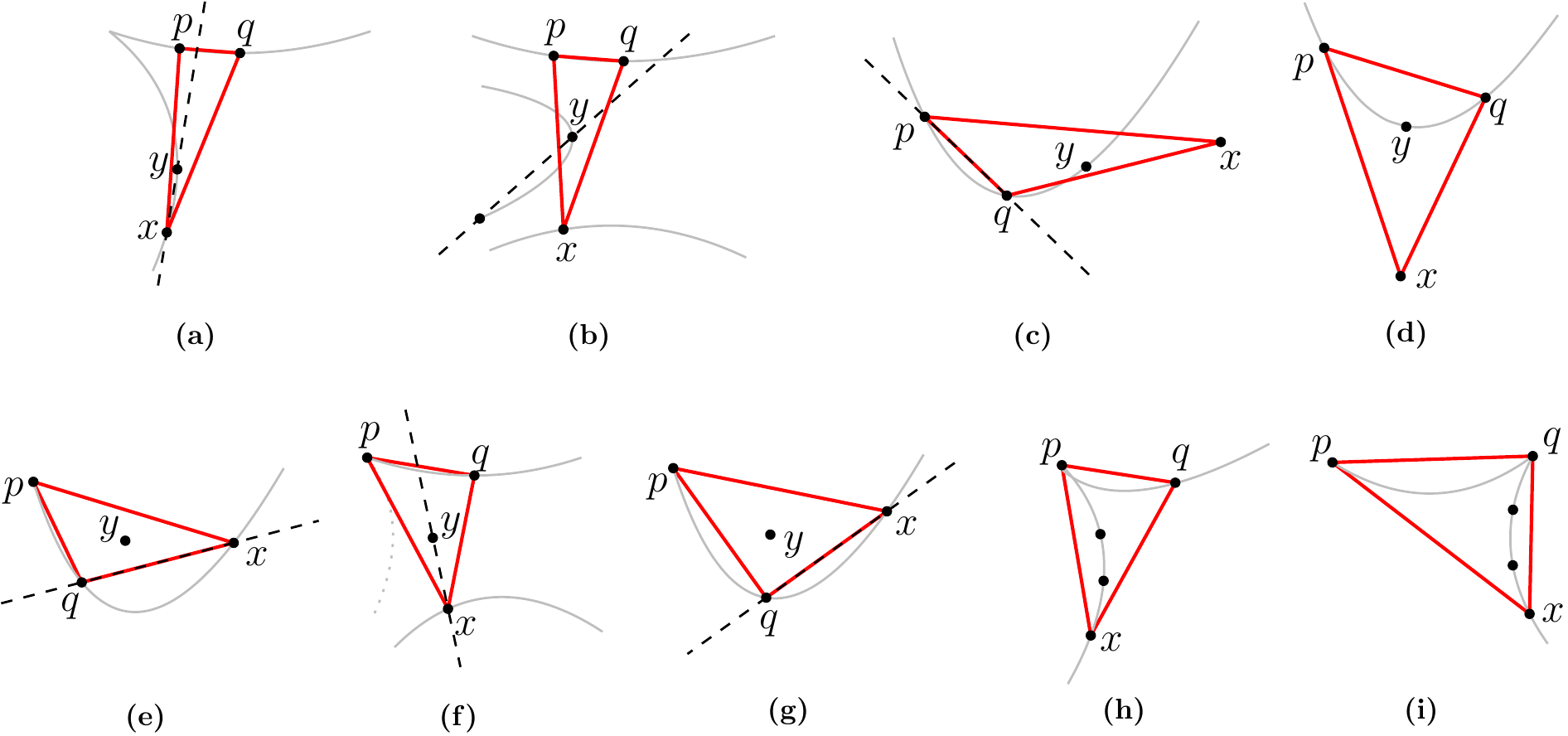}
\caption{Illustration of the proof of Lemma~\ref{app:lem:prel_gdc_1}.}
\label{app:fig:gdc_lemma}
\end{figure}

\begin{proof}
Concerning part (i), assume for the sake of contradiction that neither $p$ nor $q$ is an extreme point and let $y \in \Delta pqx$. First recall that $p$ and $q$ belong to exactly one concave chain and since they are neighbors, they belong to the same chain -- call it $C_{pq}$. Further note that $y$ is not an extreme point and hence, also belongs to a unique concave chain.

Let us first consider the case that $x$ and $y$ share a common chain $C_{xy}$. Since concave chains are in convex position, this cannot be $C_{pq}$. However, since $y$ is contained in the interior of $\Delta pqx$, the line $\ell_{xy}$ through $x$ and $y$ separates $p$ and~$q$; a contradiction to property (iv) of the chain $C_{xy}$ (see Figure~\ref{app:fig:gdc_lemma}(a)).

Consider next, that $x$ and $y$ do not share a common chain. 
If $y$ does not belong to $C_{pq}$, then neither of the points $p,q,x$ belongs to the chain $C_y$ of $y$. However, any line through $y$ and another point of $C_y$ separates these three points; a contradiction to property (iv) of $C_y$ (see Figure~\ref{app:fig:gdc_lemma}(b)).
If $y$ belongs to $C_{pq}$, then either $p$ and $q$ are not neighbors or the line $\ell_{pq}$ has $y$ and the extreme points (in $S$) of $C_{pq}$ on the same side (see Figure~\ref{app:fig:gdc_lemma}(c) and Figure~\ref{app:fig:gdc_lemma}(d)); a contradiction either way. This finishes the proof of part (i).

Concerning part (ii), first note that the extreme point $p$ belongs to two chains $C_{pq}$ and $C_{p\bar{q}}$ (the former also contains $q$ while the latter does not).
Also assume that $q$ is not an extreme point. Then we need to show that $x$ belongs to $C_{p\bar{q}}$. 

First we argue that $x$ does not belong to $C_{pq}$. Indeed, if this was the case, the point $y \in \Delta pqx$ (which does not belong to $C_{pq}$) would lie on the same side of $\ell_{qx}$ as $p$ (see Figure~\ref{app:fig:gdc_lemma}(e)); a contradiction to property (iv) of $C_{pq}$.

Second, suppose for the sake of contradiction that $x$ belongs to some other chain $C \neq C_{p\bar{q}}$. 
Again we consider the two cases where $x$ and $y$ share a common chain or not. 

First, let $x$ and $y$ share a common chain (which, as before, is not $C_{pq}$).
Since $\ell_{xy}$ separates $p$ and $q$ this is a contradiction, except $p,x,y$ lie on a common chain, which however is excluded because $x \notin C_{p\bar{q}}$ (see Figure~\ref{app:fig:gdc_lemma}(f)).

Second, let $x$ and $y$ not share a common chain. The cases that $y$ belongs to a chain that is not $C_{p\bar{q}}$ are completely analogous as before (see Figures~\ref{app:fig:gdc_lemma}(b,c,d)). Hence, it remains to consider $y \in C_{p\bar{q}}$. In this case, the line $\ell_{py}$ separates $q$ and~$x$; again a contradiction (see Figure~\ref{app:fig:gdc_lemma}(g)).

The case that $q$ is an extreme point is now completely analogous. Figures~\ref{app:fig:gdc_lemma}(h,i) summarize the possible configurations for a non-empty triangle $\Delta pqx$.\qed
\end{proof}

\subparagraph*{Combinatorial distance measure.}
In contrast to the proof for wheel sets, it may now not be possible anymore to directly increase the number of spine edges and hence, we need a more sophisticated measure.
Let $\HC$ be the spine of a spinal point set $S$ and $p,q \in S$. Further let $o \in \{\text{cw},\text{ccw}\}$ be an orientation. We define the 
\emph{distance} between $p,q$ \emph{in direction $o$}, denoted by 
$d^o(p, q)$, as the number of spine edges along $\HC$ that
lie between $p$ and $q$ in direction $o$. Furthermore, we define the \emph{distance} between $p$ and $q$ to be
\[
d(p, q) = \min \{ d^{\text{cw}}(p, q), d^{\text{ccw}}(p, q)\}.
\] 

Note that neighboring points along the spine have distance one. Using this notion we define the \emph{\weight} of an edge to be the distance between its endpoints and
the (overall) \weight of a graph on $S$ 
to be the sum of its edge weights.

\begin{figure}[t]
	\centering
	{\includegraphics[page=1,scale=0.75]{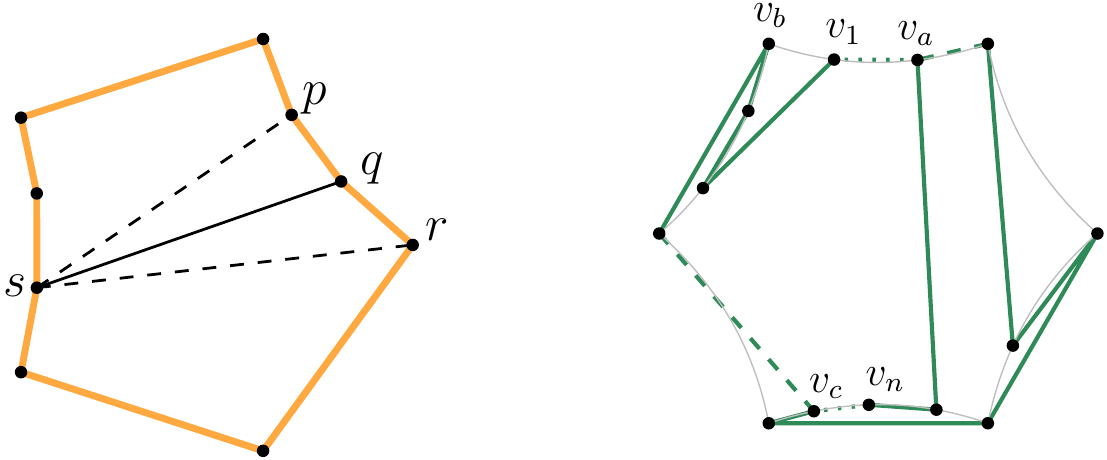}}
\caption{\label{fig:preliminary_observations}
     \emph{Left:} Illustration of Observation~\ref{obs:three_neighbors}. One of the dashed edges has smaller \weight than the solid: $d(s,q) = 4$; $d(s,p) = 4$; $d(s,r) = 3$.
     \emph{Right:} Illustration of Lemma~\ref{app:lem:GDC:neighbors}. The initial path is depicted by solid and dashed edges. Flipping the dashed edges to the dotted edges increases the number of spine edges. 
	}
\end{figure}

Our goal will be to perform \weight decreasing flips. To this end, we state two more useful preliminary results (see also Figure~\ref{fig:preliminary_observations}):

\begin{observation}\label{obs:three_neighbors}
Let $S$ be a spinal point set, 
$p, q, r$ be three neighboring points in this order (i.e., $q$ lies between $p$ and $r$), and $s \in S\setminus\{p,q,r\}$ be another point. Then $d(p,s) < d(q,s)$ or $d(r,s) < d(q,s)$ holds.
\end{observation}

Combining Observation~\ref{obs:valid_flips} and Observation~\ref{obs:three_neighbors}, it is apparent that we can perform \weight decreasing flips whenever $\Delta v_{\mya-1}v_\mya v_{\mya+1}$ and $\Delta v_1v_\mya v_{\mya+1}$ are empty. 

\begin{restatable}{lemma}{lemGDCNeighbors}\label{app:lem:GDC:neighbors}
Let $S$ be a spinal point set, $P = v_1, \ldots, v_n \in \PS$, and $v_\mya, v_\myb$ ($\mya, \myb \neq 2$) be neighbors of $v_1$ as well as $v_\myc, v_\myd$ ($\myc, \myd \neq n-1$) be neighbors of~$v_n$. If $\max(\mya, \myb) > \min(\myc, \myd)$, then the number of spine edges in $P$ can be increased by performing at most two flips, which also decrease the overall \weight of $P$.
\end{restatable}

\begin{proof}\label{lem:GDC:neighbors:proof}
Let $\max(\mya, \myb) > \min(\myc, \myd)$ and assume w.l.o.g.\ that $\mya > \myc$ holds. If $v_{\mya-1}v_\mya$ or $v_\myc v_{\myc+1}$ is not a spine edge, we can increase the number of spine edges by a single flip using Observation~\ref{obs:valid_flips}(a). Hence, we can assume these edges to be spine edges, which implies that $v_\mya v_{\mya+1}$ and $v_{\myc-1}v_\myc$ are not spine edges. 

Then we flip the spine edge $v_{\mya-1}v_\mya$ to another spine edge $v_1 v_\mya$ resulting in the path $v_{\mya-1}, \ldots, v_{\myc-1}, v_\myc, \ldots, v_1, v_\mya, \ldots, v_n$. 
Next, flipping $v_{\myc-1}v_\myc$ to $v_\myc v_n$ is valid 
and now replaces a non-spine edge by a spine edge. See Figure~\ref{fig:preliminary_observations} (right) for an illustration.\qed
\end{proof}

Note that $v_\myb$ or $v_\myd$ in Lemma~\ref{app:lem:GDC:neighbors} may not exist, if the first or last edge of $P$ is a spine edge.
Lemma~\ref{app:lem:GDC:neighbors} essentially enables us to perform \weight decreasing flips whenever the path traverses a neighbor of $v_n$ before it reached both neighbors of $v_1$.
We are now ready to prove Theorem~\ref{thm:GDC}, but briefly summarize the proof strategy from a high level perspective beforehand:

\subparagraph*{High level proof strategy.} To flip an arbitrary path $P \in \PS$ to a canonical path, we perform iterations of suitable flips such that in each iteration we either

\begin{enumerate}
	\item[(i)] increase the number of spine edges along $P$ (while not increasing the overall \weight of~$P$), or
	\item[(ii)] decrease the overall \weight of $P$ (while not 
	decreasing the number of spine edges along~$P$).
\end{enumerate}

Note that for the connectivity of the flip graph it is not necessary to guarantee the non increasing overall \weight in the first part. However, this will provide us with a better bound on the diameter of the flip graph.

\thmGDC*

\begin{proof}
Let $P~=~v_1,\dots, v_n~\in~\PS$ be a non-canonical path. As described in the high level proof strategy in Section~\ref{sec:GDC}, we show how to iteratively transform~$P$ to a canonical path by increasing the number of spine edges or decreasing its overall weight. As usual let $v_\mya$ ($\mya \neq 2$) be a neighbor of $v_1$.

We can assume, w.l.o.g., that $v_1$ and $v_n$ are not neighbors, since otherwise we can flip an arbitrary (non-spine) edge of $P$ to the spine edge $v_1v_n$ (performing a Type 2 flip), i.e., $a < n$. 
Furthermore, we can also assume w.l.o.g., that $v_{\mya-1}v_\mya$ is a spine edge, since otherwise we can flip $v_{\mya-1}v_\mya$ to the spine edge $v_1v_\mya$ (Observation~\ref{obs:valid_flips}(a)). 
This also implies that the edge $v_\mya v_{\mya+1}$, which exists because $\mya < n$, is not a spine edge, since $v_\mya$ already has the two neighbors $v_{\mya-1}$ and $v_1$.

We distinguish two cases -- $v_1v_2$ being a spine edge or not:

\begin{description}
\item[Case 1\label{app:gdc_case_1}] $v_1v_2$ is not a spine edge.

This case is easier to handle, since we are guaranteed that both neighbors $v_\mya, v_\myb$ $(\mya,\myb \neq 2)$ of $v_1$ are potential candidates to flip to. In order to apply Observation~\ref{obs:valid_flips}, we require $\Delta v_1v_\mya v_{\mya+1}$ or $\Delta v_1v_\myb v_{\myb+1}$ to be empty. So, let's consider these two subcases separately:

\begin{description}
\item[Case 1.1\label{app:gdc_case_1_1}] $\Delta v_1v_\mya v_{\mya+1}$ (or analogously $\Delta v_1v_\myb v_{\myb+1}$) is empty.

Then we apply the following two flips:
\[
\text{flip } v_\mya v_{\mya+1} \text{ to } v_1v_{\mya+1} \qquad \text{ and } \qquad \text{flip } v_1v_2 \text{ to } v_1v_\mya,
\]
where the first flip results in the path $v_\mya, \ldots, v_1, v_{\mya+1}, \ldots, v_n$ (and is valid by Observation~\ref{obs:valid_flips}(c)) and the second flip then results in the path $v_2, \ldots, v_\mya, v_1, v_{\mya+1}, \ldots, v_n$ (valid due to Observation~\ref{obs:valid_flips}(a)). The first flip replaces a non-spine edge by another non-spine edge and may increase the \weight by at most one. The second flip replaces a non-spine by a spine edge, which also decreases the \weight by at least one. 
Together, the number of spine edges increases, while the overall \weight does not increase.

\item[Case 1.2\label{app:gdc_case_1_2}] $\Delta v_1v_\mya v_{\mya+1}$ and $\Delta v_1v_\myb v_{\myb+1}$ are not empty.

Lemma~\ref{app:lem:prel_gdc_1}(i) implies that ($v_1$ or $v_\mya$) and ($v_1$ or $v_\myb$) is an extreme point.

\begin{description}
\item[Case 1.2.1\label{app:gdc_case_1_2_0a}] $v_1, v_\mya, v_\myb$ are all extreme points.

The triangles $\Delta v_1v_\mya v_{\mya+1}$ and $\Delta v_1v_\myb v_{\myb+1}$ being non-empty implies that both, $v_\mya v_{\mya+1}$ and $v_\myb v_{\myb+1}$ are outer edges (see Figure~\ref{app:fig:gdc_problematic_1}(a)). Not both of them can contain $n/2$ points under\footnote{For an outer edge $e$ (which necessarily connects two points from the same chain), we say that the points along this chain between the two endpoints of $e$ (excluding the endpoints) lie \emph{under} $e$.}
its edge. Hence, one of the two flips -- replacing $v_\mya v_{\mya+1}$ by $v_{\mya-1}v_{\mya+1}$, or $v_\myb v_{\myb+1}$ by $v_{\myb-1}v_{\myb+1}$ must be \weight decreasing, say w.l.o.g. the one for $v_\mya$. Then we apply the following flips:
\[
\text{flip } v_{\mya-1}v_\mya \text{ to } v_1v_\mya \qquad \text{ and } \qquad \text{flip } v_\mya v_{\mya+1} \text{ to } v_{\mya-1}v_{\mya+1}.
\]

The first flip is valid due to Observation~\ref{obs:valid_flips}(a) (replacing a spine edge by another spine edge) and the second is valid due to Observation~\ref{obs:valid_flips}(b) and decreases the \weight by assumption.

\item[Case 1.2.2\label{app:gdc_case_1_2_0b}] $v_1$ and $v_\mya$ are extreme points and $v_\myb$ is not (analogous with exchanged roles of $v_\mya$ and $v_\myb$).

In this case, using Lemma~\ref{app:lem:prel_gdc_1}(ii), we conclude that the triangle $\Delta v_1v_\myb v_{\myb+1}$ must be empty; a contradiction. 

\item[Case 1.2.3\label{app:gdc_case_1_2_1}] $v_1$ is an extreme point and $v_\mya, v_\myb$ are not.

By Lemma~\ref{app:lem:prel_gdc_1}(ii), the edge $v_\mya v_{\mya+1}$ must be an inner edge between the two concave chains of $v_1$. 
If $v_{\mya+1}$ is a neighbor of $v_1$ we can simply replace $v_\mya v_{\mya+1}$ by $v_1v_{\mya+1}$ (Observation~\ref{obs:valid_flips}(a)). Otherwise, the other neighbor $v_\myb$ of $v_1$ cannot be incident to an inner edge and hence, $\Delta v_1v_\myb v_{\myb+1}$ is empty; a contradiction 
(see Figure~\ref{app:fig:gdc_problematic_1}(b)).

\begin{figure}[t]
\centering
\includegraphics[scale=0.61,page=1]{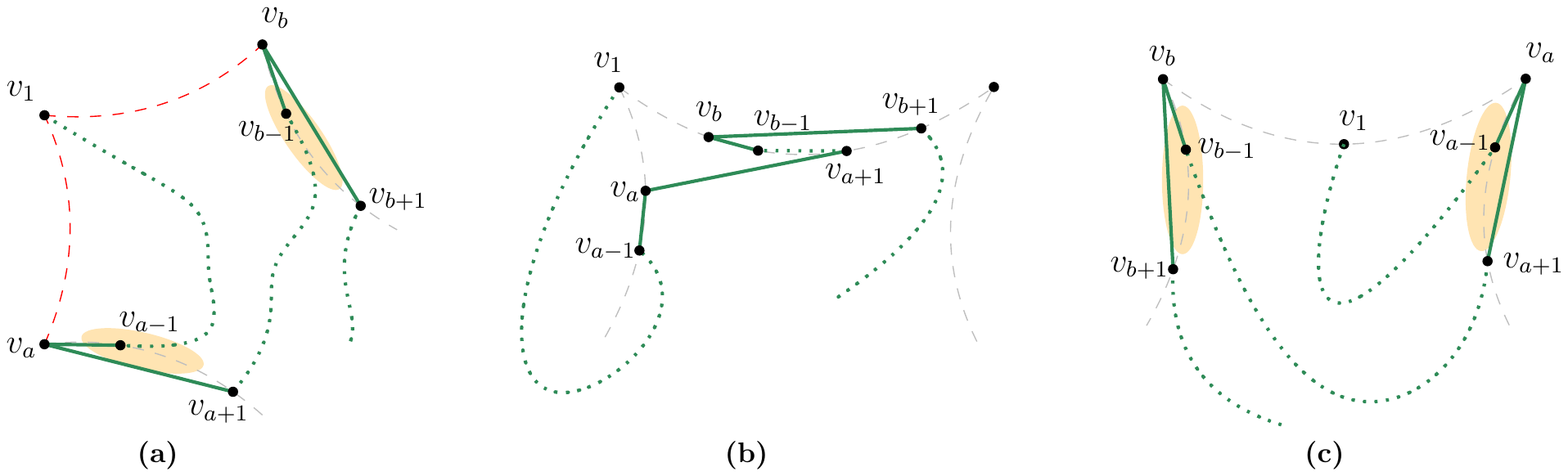}
\caption{(a) \ref{app:gdc_case_1_2_0a}. There cannot be more than $n/2$ points in each yellow region and hence, at least one of the flips $v_\mya v_{\mya+1}$ to $v_{\mya-1} v_{\mya+1}$ or $v_\myb v_{\myb+1}$ to $v_{\myb-1} v_{\myb+1}$ decreases the \weight. (b) \ref{app:gdc_case_1_2_1}. Either $\Delta v_1v_\mya v_{\mya+1}$ or $\Delta v_1v_\myb v_{\myb+1}$ must be empty. (c) \ref{app:gdc_case_1_2_2}.}
\label{app:fig:gdc_problematic_1}
\end{figure}

\item[Case 1.2.4\label{app:gdc_case_1_2_2}] $v_1$ is not an extreme point and $v_\mya, v_\myb$ are both extreme points.

This case is analogous to \ref{app:gdc_case_1_2_0a}. Again, both, $v_\mya v_{\mya+1}$ and $v_\myb v_{\myb+1}$ must be outer edges and the regions under those two edges must be disjoint (see Figure~\ref{app:fig:gdc_problematic_1}(c)). Hence, one of the two flips -- replacing $v_\mya v_{\mya+1}$ by $v_{\mya-1}v_{\mya+1}$, or $v_\myb v_{\myb+1}$ by $v_{\myb-1}v_{\myb+1}$ must be \weight decreasing, say w.l.o.g. the one for $v_\mya$. Then, as in \ref{app:gdc_case_1_2_0a}, we apply the following flips:
\[
\text{flip } v_{\mya-1}v_\mya \text{ to } v_1v_\mya \qquad \text{ and } \qquad \text{flip } v_\mya v_{\mya+1} \text{ to } v_{\mya-1}v_{\mya+1}.
\]
\end{description}
\end{description}

\item[Case 2\label{app:gdc_case_2}] $v_1v_2$ is a spine edge.

In this case we will consider $P$ from both ends $v_1$ and $v_n$. 
Our general strategy here is to first rule out some easier cases and collect all those cases where we cannot immediately make progress. For these remaining \enquote{bad} cases we consider the setting from both ends of the path, i.e., we consider all combinations of bad cases.

\begin{description}
\item[Case 2.1\label{app:gdc_case_2_1}] $v_1, v_\mya$ are not extreme points. And $v_{\mya-1}$ is also not an extreme point.
Then we either 

\vspace*{0.1cm}
\begin{itemize}
\item flip $v_\mya v_{\mya+1}$ to $v_1v_{\mya+1}$, or \vspace*{0.1cm}
\item flip $v_{\mya-1}v_\mya$ to $v_1v_\mya$ and $v_\mya v_{\mya+1}$ to $v_{\mya-1}v_{\mya+1}$.
\end{itemize}
\vspace*{0.1cm}

By Observation~\ref{obs:three_neighbors}, one of the two choices decreases the overall \weight and all flips are valid because neither of the vertices $v_1, v_\mya, v_{\mya-1}$ is an extreme point and hence, the triangles $\Delta v_1v_\mya v_{\mya+1}$ and $\Delta v_{\mya-1}v_\mya v_{\mya+1}$ are empty.

\item[Case 2.2\label{app:gdc_case_2_2}] $v_1, v_\mya$ are not extreme points. And $v_{\mya-1}$ is an extreme point.

\begin{description}

\item[Case 2.2.1\label{app:gdc_case_2_2_1}] $v_{\mya-2}$ lies on the same chain as $v_1$. 

\begin{description}
\item[Case 2.2.1.1\label{app:gdc_case_2_2_1_1}] $d(v_{\mya-2}v_{\mya-1}) = 2$.

This implies $v_{\mya-2}$ to be an extreme point and there is only one other point in $S$ that does not belong to the chain of $v_{\mya-2}$ and~$v_{\mya-1}$; see Figure~\ref{app:fig:gdc_problematic_9}(a). This, however, implies that both triangles $\Delta v_1v_\mya v_{\mya+1}$ and $\Delta v_{\mya-1} v_\mya v_{\mya+1}$ are empty and hence, using Obervation~\ref{obs:valid_flips} and Observation~\ref{obs:three_neighbors}, we can perform \weight decreasing flips.

\item[Case 2.2.1.2\label{app:gdc_case_2_2_1_2}] $d(v_{\mya-2}v_{\mya-1}) > 2$.

Then we flip $v_{\mya-2}v_{\mya-1}$ to $v_1v_{\mya-1}$, which decreases the overall \weight, since $v_1v_{\mya-1}$ has \weight 2. See Figure~\ref{app:fig:gdc_problematic_9}(b).

\end{description}

\item[Case 2.2.2\label{app:gdc_case_2_2_2}] $v_{\mya-2}$ lies on a different chain as $v_1$. 

If $v_{\mya+1}$ lies in the exterior of $\Phi(v_\mya)$, then $v_{\mya-1},v_\mya, v_1,v_{\mya+1}$ share a common chain and we can apply the flips of \ref{app:gdc_case_2_1}. 

If $v_{\mya+1}$ lies in the interior of $\Phi(v_\mya)$ this constitutes our first \enquote{bad} case (note that Observation~\ref{obs:face} implies that the subpath $v_{\mya+1},\ldots, v_n$ does not contain any extreme point). Further, we can assume $v_{\mya-2}v_{\mya-1}$ also to be a spine edge, since otherwise we flip $v_{\mya-1}v_\mya$ to $v_1v_\mya$ and are back in \ref{app:gdc_case_1}. Also we can assume $v_\mya v_{\mya+1}$ to be an inner edge towards the $v_{\mya-1}$ chain, since otherwise we can make progress (using Observations~\ref{obs:valid_flips} and \ref{obs:three_neighbors}) because both triangles $\Delta v_1v_\mya v_{\mya+1}$ and $\Delta v_{\mya-1}v_\mya v_{\mya+1}$ would be empty by Lemma~\ref{app:lem:prel_gdc_1}(ii). See Figure~\ref{app:fig:gdc_problematic_9}(c).
\end{description}

\begin{figure}[t]
\centering
\includegraphics[scale=0.6,page=9]{gdc_problematic}
\caption{(a) \ref{app:gdc_case_2_2_1_1}. (b) \ref{app:gdc_case_2_2_1_2}. Replace $v_{\mya-2}v_{\mya-1}$ by $v_1v_{\mya-1}$. (c) The first \enquote{bad} case as described in \ref{app:gdc_case_2_2_2}.}
\label{app:fig:gdc_problematic_9}
\end{figure}

\begin{figure}[t]
\centering
\includegraphics[scale=0.58,page=10]{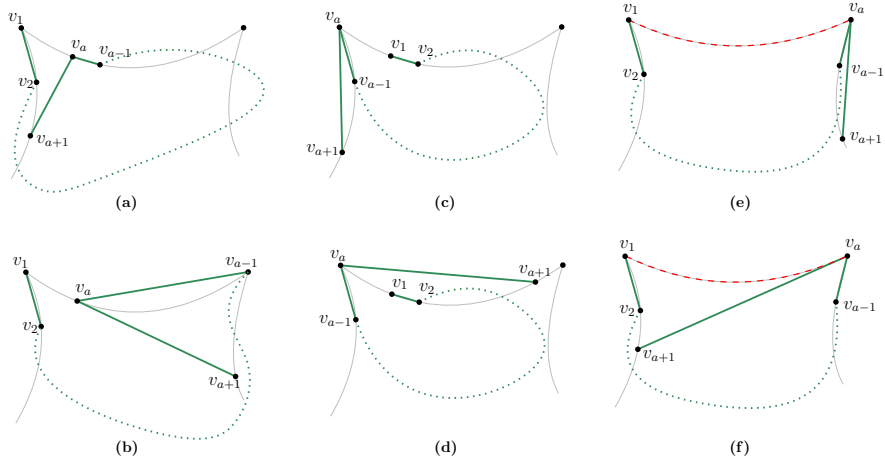}
\caption{(a,b) The \enquote{bad} cases of \ref{app:gdc_case_2_3}. Either way, $v_\mya v_{\mya+1}$ is in the interior of $\Phi(v_\mya)$. (c,d) The two \enquote{bad} cases of \ref{app:gdc_case_2_4}. The edge $v_\mya v_{\mya+1}$ is in the exterior of $\Phi(v_\mya)$. (e,f) The two \enquote{bad} cases of \ref{app:gdc_case_2_5}. The edge $v_\mya v_{\mya+1}$ can be in the interior (bottom) or exterior (top) of $\Phi(v_\mya)$. Recall that the dashed red arc shall emphasize that there is no vertex lying on this chain between the two extreme points $v_1, v_\mya$.}
\label{app:fig:gdc_problematic_10}
\end{figure}

\item[Case 2.3\label{app:gdc_case_2_3}] $v_1$ is an extreme point and $v_\mya$ is not an extreme point.

Again it suffices to consider the case that at least one of the triangles $\Delta v_1v_\mya v_{\mya+1}$ or $\Delta v_{\mya-1}v_\mya v_{\mya+1}$ is non-empty. Either way, again using Lemma~\ref{app:lem:prel_gdc_1}(ii), this can only happen, if $v_{\mya+1}$ is in the interior of $\Phi(v_\mya)$, more precisely, $v_\mya v_{\mya+1}$ is an inner edge towards one of the neighboring chains of the $v_\mya$ chain (where the one case is only relevant if $v_{\mya-1}$ is an extreme point); see Figure~\ref{app:fig:gdc_problematic_10}(a,b). This constitutes the second \enquote{bad} case. As before, there are no extreme points along the subpath $v_{\mya+1},\ldots, v_n$.

\item[Case 2.4\label{app:gdc_case_2_4}] $v_\mya$ is an extreme point and $v_1$ is not an extreme point.

If $v_\mya v_{\mya+1}$ is an inner edge both triangles $\Delta v_1v_\mya v_{\mya+1}$ and $\Delta v_{\mya-1}v_\mya v_{\mya+1}$ are empty and hence,
there are two \enquote{bad} cases to consider here, namely when $v_\mya v_{\mya+1}$ is an outer edge (see Figure~\ref{app:fig:gdc_problematic_10}(c,d)). 

\item[Case 2.5\label{app:gdc_case_2_5}] $v_1,v_\mya$ are extreme points.

Again, there are two \enquote{bad} cases to consider here, namely when $v_\mya v_{\mya+1}$ is an outer edge, or an inner edge to the $v_1,v_2$ chain (see Figure~\ref{app:fig:gdc_problematic_10}(e,f)).

\end{description}
\end{description}

Let us summarize the six \enquote{bad} cases, where we cannot immediately make progress (note that (II) comprises two cases but for the following arguments it will not be important to distinguish between them):

\begin{enumerate}
\item[(I)] $v_1, v_\mya$ are not extreme points, $v_{\mya-1}$ is an extreme point, $v_\mya v_{\mya+1}$ is an inner edge towards the $v_{\mya-1}$ chain that does not contain $v_1$, and $v_{\mya-1}v_{\mya-2}$ is a spine edge (Figure~\ref{app:fig:gdc_problematic_9}(c)).

\item[(II)] $v_1$ is an extreme point, $v_\mya$ is not an extreme point, and $v_\mya v_{\mya+1}$ is an inner edge towards a neighboring chain (in the one case $v_{\mya-1}$ must be extreme); Figure~\ref{app:fig:gdc_problematic_10}(a,b).

\item[(IIIa)] $v_1$ is not an extreme point, $v_\mya$ is an extreme point, and $v_\mya v_{\mya+1}$ is an outer edge on the $v_{\mya-1}$ chain (Figure~\ref{app:fig:gdc_problematic_10}(c)). 

\item[(IIIb)] $v_1$ is not an extreme point, $v_\mya$ is an extreme point, and $v_\mya v_{\mya+1}$ is an outer edge on the $v_1$ chain (Figure~\ref{app:fig:gdc_problematic_10}(d)). 

\item[(IVa)] $v_1,v_\mya$ are extreme points and $v_\mya v_{\mya+1}$ is an outer edge (Figure~\ref{app:fig:gdc_problematic_10}(e)).

\item[(IVb)] $v_1,v_\mya$ are extreme points and $v_\mya v_{\mya+1}$ is an inner edge to the $v_1$ chain (Figure~\ref{app:fig:gdc_problematic_10}(f)).
\end{enumerate}

In the remainder of the proof we settle these \enquote{bad} cases by arguing about both ends of the path, i.e, we consider all $\binom{6}{2} + 6 = 21$ combinations of these \enquote{bad} cases. Recall that the order at the $v_n$ end must be thought of as inverted, e.g., the analogue of $v_{\mya-1}v_\mya$ is $v_\myc v_{\myc+1}$ and the analogue of $v_\mya v_{\mya+1}$ is $v_{\myc-1} v_\myc$.

Further note that there two non-isomorphic ways two realize a combination of \enquote{bad} cases (see Figure~\ref{app:fig:gdc_problematic_7}). Furthermore, there are two ways to connect the \enquote{loose} ends of the fixed structures with each other (one case where we have $\mya < \myc$ and another where we have $\myc < \mya$). However, by Lemma~\ref{app:lem:GDC:neighbors}, only the former is of interest.

\begin{figure}[t]
\centering
\includegraphics[scale=0.75,page=7]{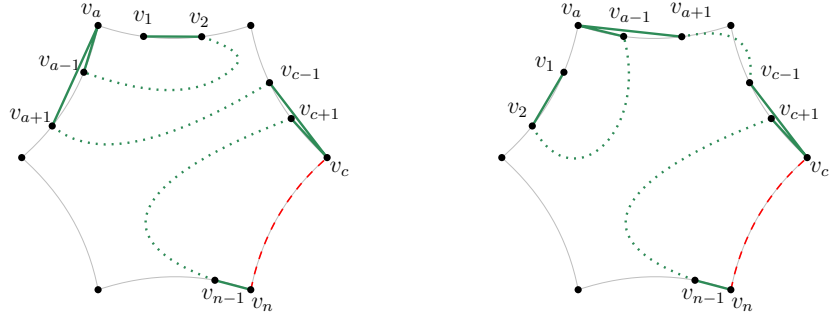}
\caption{
The two non-isomorphic ways to realize a combination of \enquote{bad} cases ((IIIa)+(IVa) is drawn here). Recall that the red dashed arc indicates that there are no more vertices on this chain.
}
\label{app:fig:gdc_problematic_7}
\end{figure}

Since we are not using the edges of the intermediate subpaths, we are mostly interested in the question whether or not a certain combination can be realized in a plane manner (in one of the two ways) or not.

Fortunately, we can immediately exclude several combinations with the following observations:
\begin{enumerate}
\item As noted above, we can assume $\mya < \myc$ and hence, no \enquote{bad} case where $v_{\mya+1}$ is in the interior of $\Phi(v_\mya)$ can be combined with a \enquote{bad} case having $v_n$ or $v_\myc$ as extreme point. Hence, we can exclude all combinations involving (I), (II), or (IVb), except (I)+(I). In fact, (I)+(I) is also not possible as the path starting from $v_1$ would first need to traverse $v_{\myc+1}$, then $v_\mya$ and then $v_\myc$, which is not possible.

\item If $v_\mya v_{\mya+1}$ is an outer edge, there is always an \enquote{easy} flip possible to either $v_1$ or $v_{\mya-1}$ (depending on which lies on the same chain as $v_\mya$ and $v_{\mya+1}$). Hence, similar to \ref{app:gdc_case_1_2_2}, we can rule out cases by showing that the regions (towards we flip $v_\mya v_{\mya+1}$ and $v_{\myc-1} v_\myc$ to) are disjoint. Also, we need to consider a subtle edge case, namely when $v_\mya v_{\mya+1}$ and $v_{\myc-1} v_\myc$ coincide. We consider the remaining six cases separately, as depicted in Figure~\ref{app:fig:gdc_problematic_8}. As depicted there, we have disjoint regions in all 6 remaining combinations except for the edge cases where the edges $v_\mya v_{\mya+1}$ and $v_{\myc-1} v_\myc$ coincide.
\end{enumerate}

\begin{figure}
\centering
\includegraphics[scale=0.65,page=8]{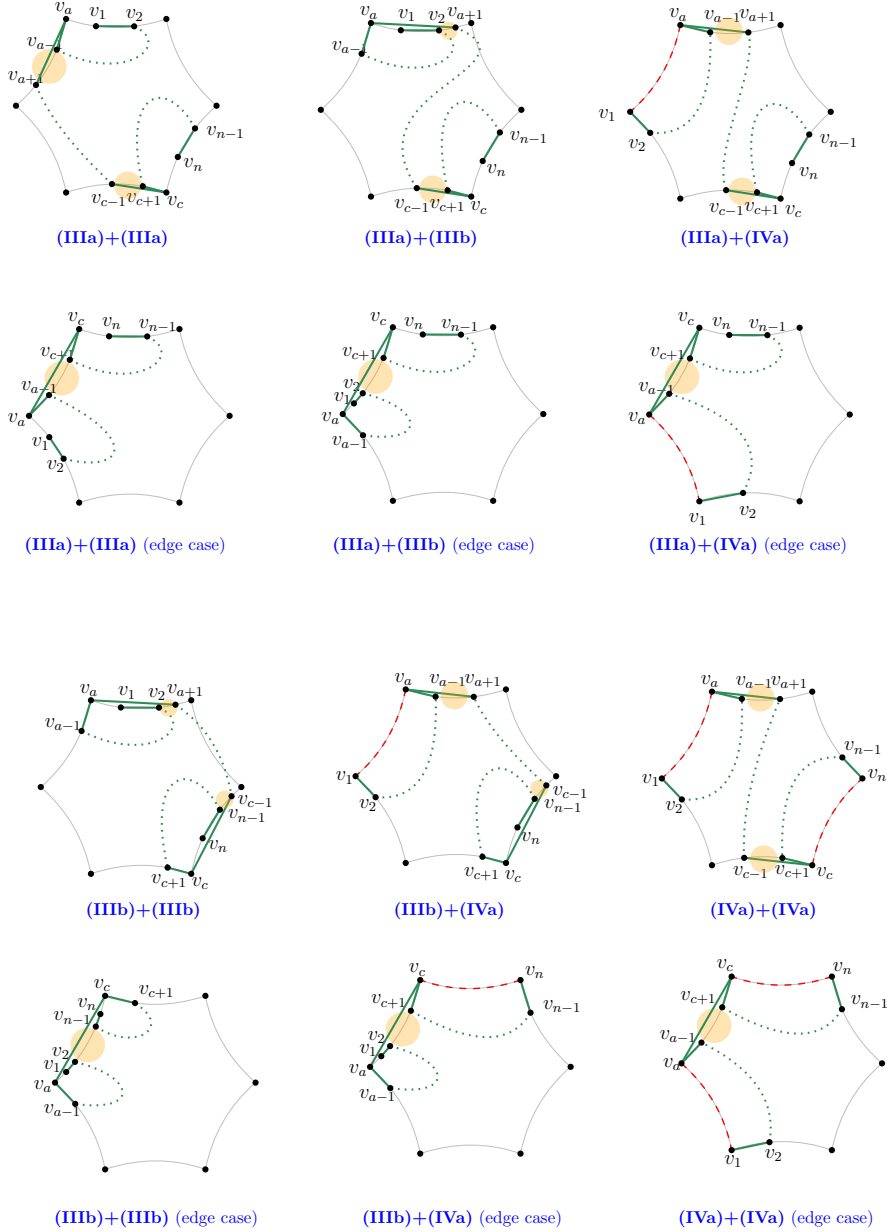}
\caption{The remaining combinations where we try to find disjoint regions each having to contain at least half of the point set. Edge cases are those cases where the edges $v_\mya v_{\mya+1}$ and $v_{\myc-1} v_\myc$ coincide.}
\label{app:fig:gdc_problematic_8}
\end{figure}

Hence, 1. and 2. rule out all possible cases, except the 6 edge cases from 2., where also these regions (containing at least $n/2$ points) coincide on the chain $CC(v_\mya, v_\myc)$. These can be resolved as follows. In all cases there exists an inner edge $v_{x-1}v_{x}$ along the path $v_2\ldots v_{\mya-1}$ that is \enquote{closest} to $v_1$ and $v_x$ lies on the chain $CC(v_\mya, v_\myc)$. By closest we mean that $v_1v_x$ does not intersect the path. Hence, flipping $v_{x-1}v_x$ to $v_1v_x$ forms a valid flip and at least one of the two choices (either from the $v_1$ end or $v_n$) decreases the \weight. This argument is analogous in all six edge cases, with the only difference that we may replace $v_1$ by $v_{\mya-1}$ by first flipping $v_{\mya-1}v_\mya$ to $v_1v_\mya$ if $v_1$ lies on the $CC(v_\mya, v_\myc)$ chain.

Recursively applying above process, we will eventually transform $P$ to a canonical path that consists only of spine edges (the only paths with minimum overall \weight). Doing the same for $Q$ and noting that any pair of canonical paths can be transformed into each other by a single flip, the connectedness of the flip graph follows.

Concerning the required number of flips, note that any edge has \weight at most $\frac{n}{2}-1$ and the path has $n-1$ edges. Hence, the total number of iterations to transform $P$ into a canonical path is at most
\[
\left(\left(n-1\right)\cdot\left(\frac{n}{2} - 2\right) + (n-1)\right) \in O(n^2)
\]
Furthermore, any iteration requires at most two flips and hence, the total number of flips to transform $P$ into $Q$ is still in $O(n^2)$.\qed
\end{proof}

\section{Conclusion}

In this paper, we made progress towards a positive answer of Conjecture~\ref{conj:main}, though it still remains open in general. We approached Conjecture~\ref{conj:main} from two directions and believe that Conjecture~\ref{conj:startedge} might be easier to tackle, e.g. for an inductive approach. 
For all our results we used only Type 1 and Type 2 flips (which can be simulated by Type 1 flips). It is an intriguing question whether Type 3 flips are necessary at all.

Concerning the approach of special classes of point sets, of course one can try to further adapt the ideas to other classes. Most of our results hold for the setting of spinal point sets; the main obstacle that remains in order to show flip connectivity in the \ref{item:p1} setting would be to adapt Lemma~\ref{app:lem:prel_gdc_1}. A proof for general point sets, however, seems elusive at the moment.

Lastly, there are several other directions for further research conceivable, e.g. considering simple drawings (or other types of drawings) instead of straight-line drawings.

\bibliography{bibliography}

\appendix

\section{Only convex point sets fulfill the empty triangle property}\label{app:sec:extreme_case}

\begin{lemma}\label{lem:extreme_case}
Let $S$ be a point set (in general position) such that there is a plane (not necessarily uncrossed) spanning cycle $\HC$ on $S$ and any triangle using an edge of $\HC$ is empty (i.e., does not contain a point of $S$). Then $S$ is in convex position.
\end{lemma}

\begin{proof}
Assume, for the sake of contradiction, that $S$ is not in convex position, and let $\HC$ be a plane spanning cycle (defining a non-convex simple polygon). It is well-known that any non-convex polygon has at least one reflex interior angle, say at vertex $r$ and further let $s$ be a neighbor of $r$ (along $\HC$). Cearly, the line through $rs$ intersects $\HC$ in an edge $pq$ and hence, either $\Delta pqr$ contains $s$ or $\Delta pqs$ contains $r$; a contradiction either way.
\end{proof}

\end{document}